\documentclass[submission,copyright,creativecommons]{eptcs}%a4paper]{llncs}

\usepackage{makeidx}  % allows for indexgeneration
\usepackage[utf8]{inputenc}
\usepackage{graphicx}
\usepackage{amssymb}
\usepackage{tikz}
\usepackage{tikz-uml}
\usepackage{mathrsfs,amssymb}
\usepackage{amsmath}
\usepackage{xspace}
\usepackage{float}
\usepackage{bbm} 
\usepackage{hyperref}
\usepackage{proof} 
\usepackage{listings}
\usepackage{pf}
\newcommand{\isabelle}{\textsc{Isabelle}\xspace}

\newcommand{\maude}{\textsc{Maude}\xspace}
\newcommand{\coq}{\textsc{Coq}\xspace}
\newcommand{\coqformde}{\textsc{Coq4MDE}\xspace}
\newcommand{\reuseware}{\textsc{ReuseWare}\xspace}

\newcommand{\model}{\texttt{Model}\xspace}

\newcommand{\metamodel}{\texttt{MetaModel}\xspace}
\newcommand{\mde}{\texttt{MDE}\xspace}
\newcommand{\mof}{\texttt{MOF}\xspace}

 \newcommand{\isc}{\texttt{ISC}\xspace}
 \newcommand{\uml}{\texttt{UML}\xspace}
 \newcommand{\ocl}{\texttt{OCL}\xspace}

 \newcommand{\java}{\texttt{Java}\xspace}
 \newcommand{\bind}{\texttt{bind}\xspace}
 \newcommand{\extend}{\texttt{extend}\xspace}
\newtheorem{theorem}{Theorem}
\newtheorem{definition}{Definition}
\newtheorem{lemma}{Lemma}

\usepackage{breakurl}             % Not needed if you use pdflatex only.

\title{Correct-by-construction model composition\\
Application to the Invasive Software Composition method}

\author{Mounira Kezadri Hamiaz
\institute{Université de Toulouse, IRIT, France}
\email{mounira.kezadri@enseeiht.fr}
\and
Benoit Combemale
\institute{ Université de Rennes 1, IRISA, France}
\email{benoit.combemale@irisa.fr}
\and
Marc Pantel
\institute{Université de Toulouse, IRIT, France}
\email{marc.pantel@enseeiht.fr}
\and
Xavier Thirioux
\institute{Université de Toulouse, IRIT, France}
\email{xavier.thirioux@enseeiht.fr}
}
\def\titlerunning{Correct-by-construction model composition: Application to Invasive Software Composition}
\def\authorrunning{M. Kezadri, M. Pantel, B. Combemale \& X. Thirioux}

\begin{document}
\maketitle

\tikzset{node distance=3cm, auto}

\maketitle

\begin{abstract}
Composition technologies improve reuse in the development of
large-scale complex systems. Safety critical systems require intensive
validation and verification activities. These activities should be
compositional in order to reduce the amount of residual verification activities
that must be conducted on the composite in addition to the ones
conducted on each components. In order to ensure the correctness of
compositional verification and assess the minimality of the residual
verification, the contribution proposes to use formal specification
and verification at the composition operator level. A first experiment
was conducted in \cite{kezadri2011proof} using proof assistants to
formalize the generic composition technology \isc and prove that type
checking was compositional. This contribution extends our early work
to handle full model conformance and study the mandatory residual
verification. It shows that \isc operators are not fully compositional
with respect to conformance and provides the minimal preconditions on
the operators mandatory to ensure compositional conformance. The
appropriate operators from \isc (especially \bind) have been
implemented in the \coqformde framework that provides a full
implementation of \mof in the \coq proof assistant. Expected
properties, respectively residual verification, are expressed as post,
respectfully pre, conditions for the composition operators. The
correctness of the compositional verification is proven in \coq.
%\keywords{Basic Composition operators, Compositional verification, ISC, MDE, MOF}
\end{abstract}

\section{Introduction}

Composition technologies improve reuse in the development of
large-scale complex systems. Safety critical systems require intensive
validation and verification activities. These activities should be
compositional in order to reduce the amount of residual activities
that must be conducted on the composite in addition to the ones
conducted on each components. In order to ensure the correctness of
compositional verification and assess the minimality of the residual
verification, the contribution proposes to use formal specification
and verification at the composition operator level. 

A first experiment was conducted in \cite{kezadri2011proof} using
proof assistants to formalize the generic composition technology \isc
\cite{assmann2003invasive} and especially the \bind and \extend
operators. This generic composition method enables to enrich the
models to express composition interfaces and to assemble the generated
components using some composition operators. Type checking for models
based on metamodels was proved to be compositional for these
operators. However, the implementation of operators in \isc does not
take into account other semantics properties for the conformance
relation for metamodels and inconsistent models can be generated.

This contribution extends our early work to handle full
model conformance and study the mandatory residual verification. It
shows that ISC operators are not fully compositional with respect to
conformance and provides the minimal preconditions on the operators
mandatory to ensure compositional conformance. The appropriate
operators from \isc (especially \bind) have been implemented in the
\coqformde framework that provides a full implementation of \mof in the
\coq proof assistant. Expected properties, respectively residual
verification, are expressed as post, respectively pre, conditions for
the composition operators. The correctness of the compositional
verification is proven in \coq.

This paper focuses on an evolution of the \isc operators (especially
the \bind operator) to correct the inconsistencies in the first
implementation that allowed to build model compositions that do not
conform to the composite metamodel. It also gives the verification for
some generic semantics properties of the \mof metamodel conformance
relation \cite{omg2011mof}.

This first section has given a short introduction. The second section presents the required notions about \coqformde. To motivate the
evolution of the \isc operators and the associated proofs, the third
section gives an example of an inconsistent metamodel assembled by the
\reuseware\footnote{\url{http://www.reuseware.org}} \cite{henriksson2008extending} \cite{Heidenreich2009} plugin from
consistent metamodels. The fourth section first discusses the formalization
of the \bind operator, then presents some preconditions for the
conformance verification and the associated proofs. The fifth section
presents some related work. Finally, the last section concludes and
gives some perspectives.

\section{Coq4MDE}\label{Coq4MDE}
This section gives the main insight of our \mde framework \coqformde,
derived from \cite{towers07}.  It defines principally the notions of \texttt{Model}
and \texttt{MetaModel}.

In our framework, the concept of metamodel is not a specialization
of the concept of model.  A model is the instance level and a
metamodel is a modeling language used to define
models. Both are formally defined in the following
way.  Let us consider two sets of labels: $\texttt{Classes}$,
respectively $\texttt{References}$, represents the set of all possible
class, respectively reference labels.  Then, let us consider
instances of such classes, the set $\texttt{Objects}$ of object
labels. $\texttt{References}$ includes a specific $inh$ label used to
specify the inheritance relation.  In the next sections, we will
elide the word label and directly talk about classes, references
and objects.
\begin{definition}[Model]
Let $\mathscr{C} \subseteq \texttt{Classes}$ be a set of classes. \\
Let $\mathscr{R} \subseteq \{ \langle c_1, r, c_2\rangle ~|~ c_1, c_2
\in \mathscr{C}, r \in\, $\texttt{Ref\-erences}$\}$\footnote{$\langle c_1,
  c_2, r\rangle$ in the \coq code is denoted here for simplification as:  $\langle
  c_1, r, c_2\rangle$.} be a set of references between classes.

A \texttt{Model} over $\mathscr{C}$ and $\mathscr{R}$, written $\langle MV,ME\rangle \in Model(\mathscr{C},\mathscr{R})$ is
a multigraph built over a finite set $MV$ of typed object vertices and a finite set $ME$\footnote{$\langle \langle o_1, c_1\rangle,r,\langle o_2, c_2\rangle\rangle$ is denoted in the \coq code as: $\langle \langle o_1, c_1\rangle,\langle o_2, c_2\rangle, r\rangle \rangle$.}  of reference edges such that:
\[
\begin{array}{l}
MV \subseteq \left\{ \langle o, c\rangle  ~|~ o\in \texttt{Objects}, c\in \mathscr{C} \right\}\\
ME \subseteq \left\{ 
  \begin{array}{l | l}
    \langle \langle o_1, c_1\rangle,r,\langle o_2, c_2\rangle\rangle &
    \begin{array}{l}
      \langle o_1, c_1\rangle, \langle o_2, c_2\rangle \in MV,
      \langle c_1, r, c_2 \rangle \in \mathscr{R}
    \end{array}
  \end{array}
\right\}
\end{array}
\]
\end{definition}
Note that, in case of inheritance, the same object label will be used
several times in the same model graph. It will be associated to the
different classes in the inheritance hierarchy going from a root to
the class used to create the object. This label reuse encodes the
inheritance polymorphism, a key aspect of most OO languages.
Inheritance is represented in the metamodel with a special reference 
called $inh$. The \texttt{subClass} property is presented in the Section~\ref{V}.

\begin{definition}[Metamodel]
  A \texttt{MetaModel} is a multigraph representing classes as vertices and references as edges as
  well as semantic properties over instantiation of classes and
  references.  It is represented as a pair composed of a multigraph
  $(MMV,MME)$ built over a finite set $MMV$ of vertices and a
  finite set $MME$ of edges, and of a predicate
  over models representing the semantic properties.
  
  A \texttt{MetaModel} is a pair $\langle (MMV,MME),conformsTo\rangle$ such that:
  \[
  \begin{array}{l}
    MMV \subseteq \texttt{Classes}\\
    MME \subseteq \{ \langle c_1, r, c_2\rangle ~|~ c_1, c_2 \in MMV, r \in \texttt{References}\}\\
    conformsTo : Model(MMV,MME) \rightarrow Bool
  \end{array}
  \]
\end{definition}
Given one \texttt{Model} $M$ and one \texttt{MetaModel} $MM$, we can check conformance using
the $conformsTo$ predicate embedded in $MM$.
It identifies the set of valid models with respect to a metamodel. 

In a prospect to construct a formal framework for model composition,
we extend the previous \mde framework to formalize and prove the
properties preservation for the \isc basic composition operators implemented inside the \reuseware framework. 

\section{An example of inconsistent metamodel generated by \reuseware}
\isc is a generic technology for extending a DSML with model composition facilities. Its first version was defined to compose \java programs and was implemented in the COMPOST system\footnote{http://www.the-compost-system.org}. A universal extension called U-ISC was proposed in \cite{Heidenreich2009}, this technique deals with textual components that can be described using context-free grammars and then the fragments are represented as trees. The method as presented considers tree merging for the composition. Recently, in order to deal with graphical languages the method was extended to support typed graphs in \cite{henriksson2008extending}, this method was implemented in the \reuseware framework. This last implementation is consistent with the description of models as graphs in our \coqformde framework.

Using the \reuseware plugin, the composition of the two models
presented in Figures~\ref{m1} and \ref{m2} by the composition program
presented in Figure~\ref{cp} generates the model presented in
Figure~\ref{cr}.

\begin{figure}[H]
\vspace{-3ex}
\begin{minipage}{0.55\linewidth}
 \centering
\includegraphics[width=0.9\textwidth]{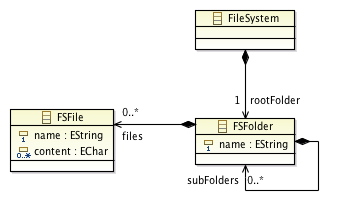}
 \vspace{-2ex}
 \caption{The advice model}
 \label{m1}
 \end{minipage}
\begin{minipage}{0.55\linewidth}
\centering
\includegraphics[width=0.4\textwidth]{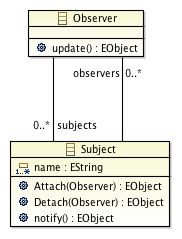}
 \vspace{-2ex}
\caption{The observer model}
\label{m2}
\end{minipage}
\vspace{-3ex}
 \end{figure}
This example is presented in~\cite{johannes2010component} and is accessible with the \reuseware Eclipse plugin applications\footnote{\url{http://www.reuseware.org/index.php/Reuseware_Aspect_Weaving}}. We slightly modified the observer model by adding an attribute \texttt{name} having as minimal multiplicity~$1$ and as maximal multiplicity~$*$ (see Figure~\ref{m2}) to illustrate the issue with the result of the composition (see Figure~\ref{cr}).
\begin{figure}[H]
\vspace{-3ex}
\begin{minipage}{0.55\linewidth}
 \centering
\includegraphics[width=0.85\textwidth]{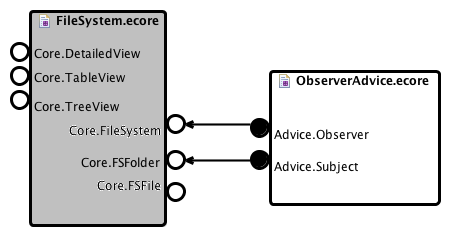}
 \vspace{-2ex}
 \caption{The composition program}
 \label{cp}
 \end{minipage}
\begin{minipage}{0.55\linewidth}
\centering
\includegraphics[width=0.8\textwidth]{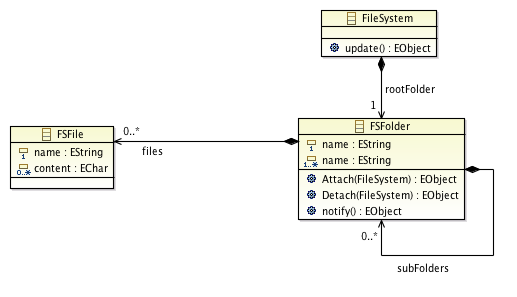}
  \vspace{-2ex}
 \caption{The composition result}
\label{cr}
\end{minipage}
\vspace{-3ex}
 \end{figure}
 
%\begin{figure}[H]
%  \centering
%  \vspace{-3ex}
%\includegraphics[width=0.45\textwidth]{Images/fsObserver.png}
% \vspace{-2ex}
% \caption{The composition program}
% \vspace{-3ex}
% \label{cp}
%\end{figure}
The composition program shown in Figure~\ref{cp} describes the links between the variation and reference points and aims to implement the class weaving for the two metamodels.
%\begin{figure}[H]
%  \centering
% \vspace{-3ex}
%\includegraphics[width=0.6\textwidth]{Images/FileSystemC.png}
%  \vspace{-2ex}
% \caption{The composition result}
%   \vspace{-2ex}
% \label{cr}
%\end{figure}

In the model obtained by composition, the class FSFolder has two attributes \texttt{name} with different multiplicities (1 and 1 .. *). This generates ambiguities and the metamodel is clearly inconsistent.

Our approach for the verification allows to detect and avoid this kind of inconsistencies by considering the metamodel semantics properties. The fact that one attribute must have a single value for the minimum and maximum multiplicities is a semantics property represented with the attributes $lower$ and $upper$ for the class $Property$ (see Figure~\ref{fig:mof})
\begin{figure}[H]
  \centering
 \vspace{-3ex}
\includegraphics[width=0.7\textwidth]{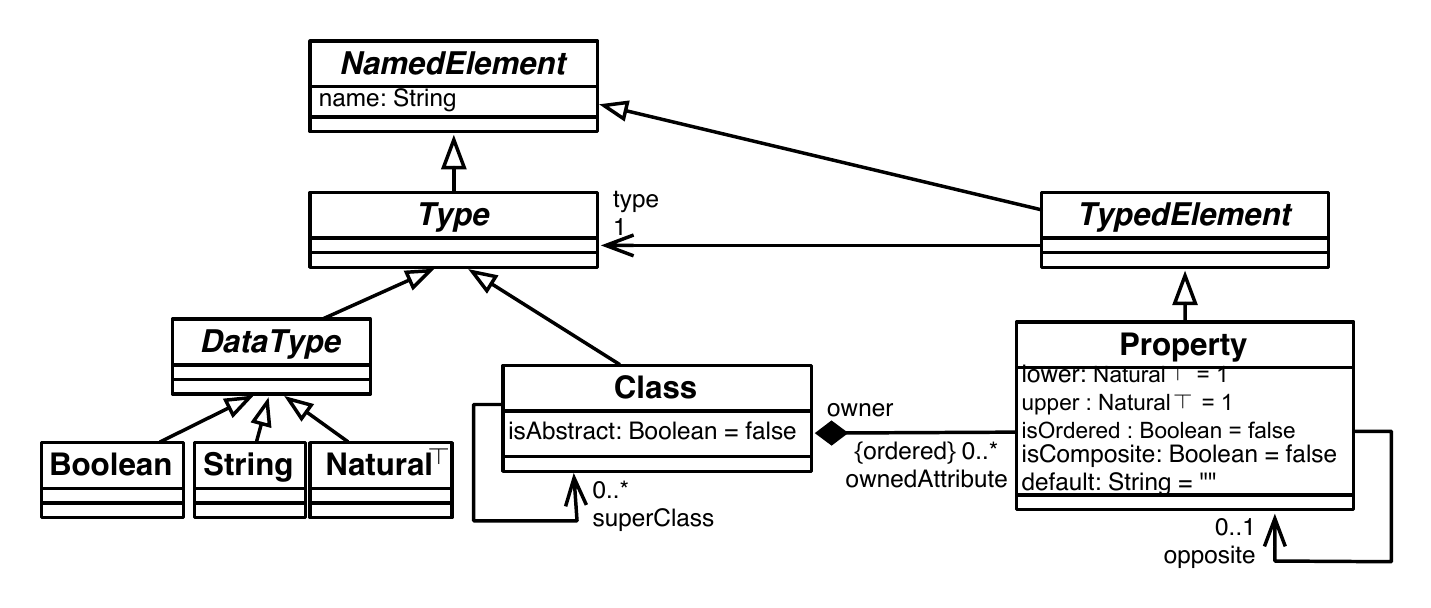}
  \vspace{-2ex}
 \caption{The \mof metamodel}
 \label{fig:mof}
 \vspace{-2ex}
\end{figure}

We consider our metamodels as models conforming to \mof (represented in Figure~\ref{fig:mof} as a metamodel), then we verify the conformance properties in relation with this metamodel. We show in the following the verifications of this kind of properties for the \isc method \bind operator.

\section{The verifications}\label{V}
The \bind operator formalized in~\cite{kezadri2011proof} enables for two models $M$ and $M'$ to replace a model's element $\mathtt{b}$ from the model $M$ referenced by a variation point by another model's element $\mathtt{b'}$ from the model $M'$ referenced by a reference point. The two model's elements $b$ and $b'$ must have the same type. This operator as presented in~\cite{kezadri2011proof} is proved compositional for typing but can generate inconsistencies on the resulted models with respect to conformance. The predicate \texttt{InstanceOf} is used to check that
all objects and links of a model are instance of classes and
references in a metamodel.
\[
\begin{array}{l}
InstanceOf(\langle \langle \mathtt{MV}, \mathtt{ME} \rangle , \langle \langle \mathtt{MMV}, \mathtt{MME} ,conformsTo \rangle \rangle  \rangle)\triangleq\\
\qquad \forall~\langle \mathtt{o},\mathtt{c} \rangle \in \mathtt{MV}, \mathtt{c} \in \mathtt{MMV}
 \wedge~\forall \langle \langle \mathtt{o},\mathtt{c} \rangle , \mathtt{r}, \langle \mathtt{o'},\mathtt{c'} \rangle \rangle \in \mathtt{ME},~ \langle \mathtt{c},\mathtt{r},\mathtt{c'} \rangle \in \mathtt{MME}\\
\end{array}
\]
Then, this predicate is used to verify that using two components
instance of \texttt{MM}, the component resulting from the application
of the \texttt{bind} operator is also instance of \texttt{MM}.

Let consider now the inheritance property represented using the relation $superClass$ in Figure~\ref{fig:mof}. This property is formally represented in \coqformde with a special reference called $inh$. The property \texttt{subClass} states that $c_2$ is a direct
  subclass of $c_1$ in the model $\langle MV, ME\rangle$.
 \[
\begin{array}{l}
 subClass(c_1, c_2\in \texttt{Classes}, \langle MV, ME\rangle) \triangleq 
\forall o~\in \texttt{Objects}, \langle o: c_2 \rangle \in MV \Rightarrow \langle \langle o: c_2\rangle, inh, \langle o: c_1\rangle\rangle \in ME
\end{array}
\]

In Figure~\ref{fig:inhNP}, we show that the \bind operator generates inconsistencies concerning the inheritance. In this example, we apply on the model the \bind operator with as parameters ($c:C$) and ($c':C'$), so that replaces the model's element ($c:C$) by ($c':C'$). The condition for this operator is that $C$ is equal to $C'$, this preserves the type safety but generates problems with the inheritance. The cause is that this replacement does not preserve the label reuse used to implement the inheritance and discussed in the Section~\ref{Coq4MDE}.

\begin{figure}[h]
 \vspace{-1ex}
\begin{center}
\scalebox{0.6}{
 	\begin{tikzpicture} 
 	
	\tikzstyle{hook}=[thick,dashed, fill=gray!20] 	
    \tikzstyle{prototype}=[thick, fill=black!20] 	
 	\draw [hook] (1.2,0) circle (0.2);
 	\draw [prototype] (2.8,0) circle (0.2);
 	\draw (-1,-1) -- (-1,5);
 	\draw (-1,5) -- (1,5);
 	\draw (1,5) -- (1,-1);
 	\draw (-1,-1) -- (1,-1);
 	
 	\draw (3,-1) -- (3,1);
 	\draw (3,1) -- (5,1);
 	\draw (5,1) -- (5,-1);
 	\draw (3,-1) -- (5,-1);
 	
 	\draw (13,-1) -- (13,5);
 	\draw (13,5) -- (15,5);
 	\draw (15,5) -- (15,-1);
 	\draw (13,-1) -- (15,-1);
 	
 	\draw (0.5,0) -- (1,0);
 	 	
 	\draw (1.4,0) -- (2.6,0);
 	
	\draw (3,0) -- (3.5,0); 	
 	
	\umlemptyclass{c: C}
	\umlemptyclass[x=4]{c': C'}
	\umlemptyclass[y=2]{c: Csup}
	\umlemptyclass[y=4]{c: Cssup}
	\umluniassoc[arg=inh, pos=0.3, align=right, name=inh]{c: C}{c: Csup} 
	\umluniassoc[arg=inh, pos=0.3, align=right, name=inh]{c: Csup}{c: Cssup}
	
	\umlemptyclass[x=14]{c': C'}
	\umlemptyclass[x=14, y=2]{c: Csup} 
	\umlemptyclass[x=14, y=4]{c: Cssup}
	\umluniassoc[arg=inh, pos=0.3, align=right, name=inh]{c': C'}{c: Csup}
	\umluniassoc[arg=inh, pos=0.3, align=right, name=inh]{c: Csup}{c: Cssup}
	
	\draw (-0.3,-1.5) node[right] {$M_1$};
	\draw (3.7,-1.5) node[right] {$M_2$};		
		
	\draw [thick, ->] (6,1) -- (12,1);
	\draw (7,1.5) node[right] {$bind~(c,C)~(c',C')~M_1~M_2$};
	\end{tikzpicture}
	}
\end{center}
\vspace{-2ex}
 	\caption{Inconsistency (concerning the inheritance) generated by the \bind operator}
 	\label{fig:inhNP}
 \vspace{-1ex}
\end{figure}
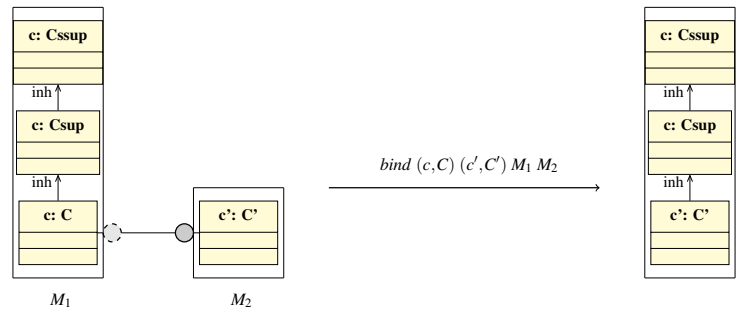

To correct this inconsistency, we slightly modified this operator. The new operator changes the name of all elements named $c$ by $c'$, the type of each element remains unchanged. We prove then the preservation of the type safety, the inheritance and other \mof properties by this operator. The \bind operator is also extended to a recursive form to support several variation and reference points as mentioned in Figure~\ref{FigSMExp}.

\begin{figure}[h]
  \centering
 \includegraphics[width=0.7\textwidth]{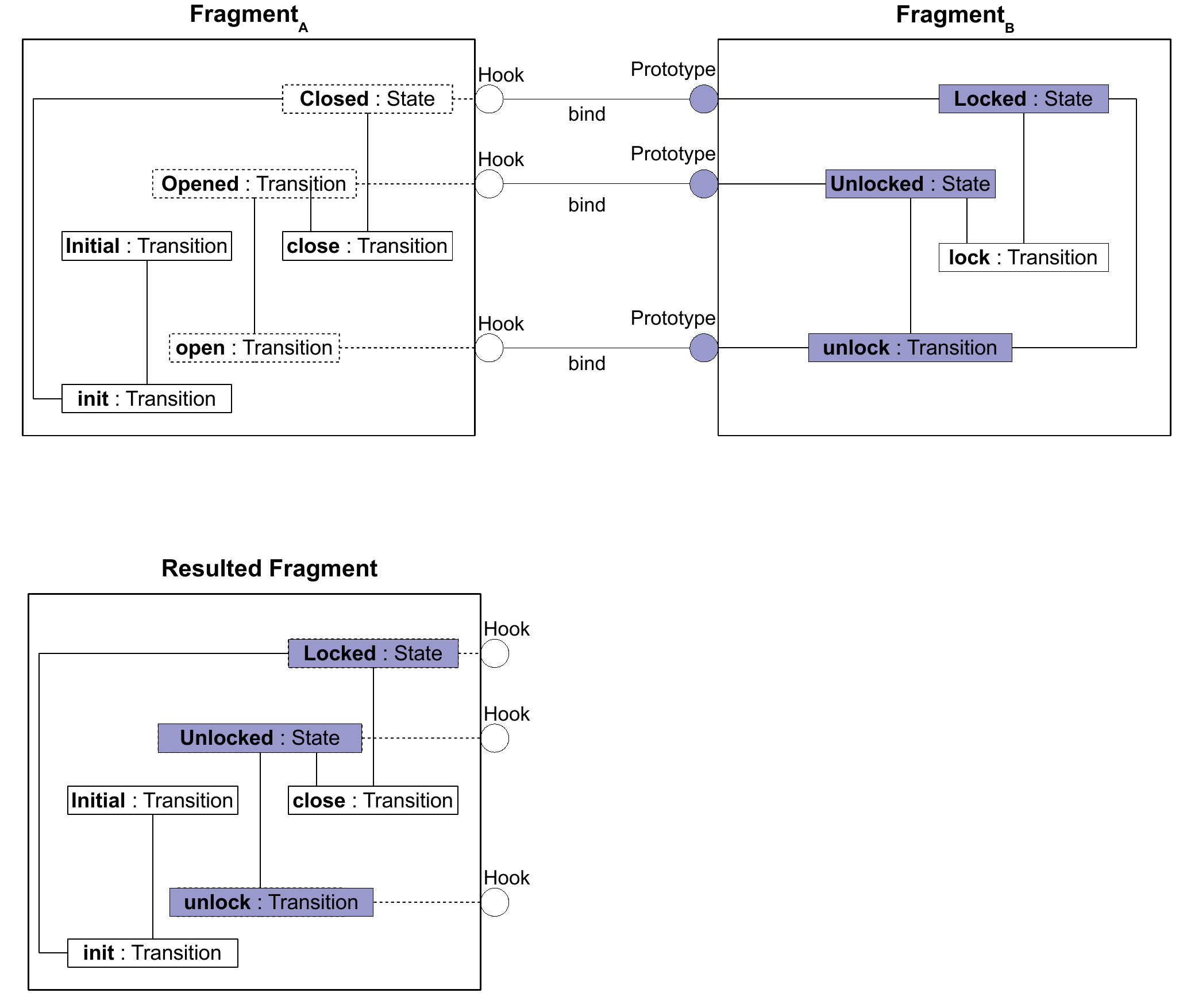}%SMExample.jpg}
  \vspace{-2ex}
 \caption{The \bind of several variation points operator}
  \vspace{-2ex}
  \label{FigSMExp}
\end{figure}

We reuse $InstanceOf$ predicate to prove the type safety for the new $bind$ operator using the theorem~\ref{ValidBind}\footnote{\url{http://coq4mde.enseeiht.fr/FormalMDE/Bind2M_Verif.html\#ValidBind}}
(\texttt{ValidBind}) for any two models $\mathtt{M_1}$ and $\mathtt{M_2}$ and any models' elements $\mathtt{o_1}$ and $\mathtt{o_2}$.
\begin{theorem}\emph{(ValidBind)}\label{ValidBind}\\
$~~~~~~~InstanceOf~  (\mathtt{M_1}, \mathtt{MM}) ~ \wedge ~InstanceOf ~ (\mathtt{M_2}, \mathtt{MM}) 
\rightarrow InstanceOf ~ ((bind  ~ \mathtt{o_1}~\mathtt{o_2}~\mathtt{M_1}  ~ \mathtt{M_2}), \mathtt{MM})$
\end{theorem}

\begin{proof}

\pfsketch{
We suppose that the two models $M_1$ and $M_2$ are instance of the metamodel $MM$  and we prove that the model obtained by applying the $bind$ operator on the two models using any two model's elements $o_1$ and $o_2$ is also instance of the metamodel $MM$. We verify first that $o_1$ is a $Hook$ for the model $M_1$ and $o_2$ is a $Prototype$ for the model $M_2$ and that $o_1$ and $o_2$ have the same type otherwise the $bind$ returns the model $M_1$ and the proof is trivial. In case of $o_1$ is a $Hook$ and $o_2$ is a $Prototype$ and the two model's elements have the same types, we show that the $bind$ does not change the types of the vertices and edges and so preserves the type safety (some details of the proof are given as an appendix).
}
\end{proof}

We developed an elegant way to prove that the basic composition operators preserve the conformance regarding the semantics properties of metamodel (other than typing). We used this method to take into account some semantics properties of the \mof metametamodel. This approach prevents us trying to extract the properties from the initial metamodel which is not obvious to do. The complexity is linked to the fact that the \texttt{conformsTo} predicate is defined in a generic way to support any kind of properties on the metamodel. The idea is to ensure that each elementary property verified on the initial models is also verified on the result of the application of the composition operator. So, if the initial models are conform to some metamodel, the resulting model is consistent with the same metamodel.
The basic semantics properties considered are: inheritance (\texttt{subClass}), abstract classes (\texttt{isAbstract}), multiplicities (\texttt{lower} and \texttt{upper}), the opposite references (\texttt{isOpposite}) and composite references (\texttt{areComposite}).

In follows, we present for each property, the theorem that proves the preservation for the \bind operator and the link to the complete \coq proof.

\subsection{The verification of some \mof properties}
We show that inheritance, abstract classes, multiplicities, opposite and composite references are preserved by the \bind operator. We note also that the proofs of these properties require in some cases additional preconditions that represent the residual verification activities when composing verified models.

\paragraph{The subClass property:} The theorem~\ref{BindSubClassPreserved} (\texttt{BindSubClassPreserved}) shows\footnote{\url{http://coq4mde.enseeiht.fr/FormalMDE/Bind2M_Verif.html\#Bind2MSCP}} that the inheritance is preserved by the \bind operator. So, for all classes $c_1$ and $c_2$ and for all model's elements $o_1$ and $o_2$, if $c_1$ is a \texttt{subClass} of $c_2$ in two models $M_1$ and $M_2$, then $c_1$ is a \texttt{subClass} of $c_2$ in the resulting model from (\bind $o_1$ $o_2$ $M_1$ $M_2$).
\begin{theorem}\emph{(BindSubClassPreserved)}\label{BindSubClassPreserved}\\
$\begin{array}{l}
\quad \forall~\mathtt{M_1}~\mathtt{M_2} \in \texttt{Model},~\mathtt{c_1}~\mathtt{c_2} \in \texttt{Classes},~ \mathtt{o_1}~\mathtt{o_2} \in \texttt{Objects},\\
\qquad (subClass~\mathtt{c_1}~\mathtt{c_2}~\mathtt{M_1}) \wedge~(subClass~\mathtt{c_1}~\mathtt{c_2}~\mathtt{M_2})
\rightarrow subClass ~\mathtt{c_1}~\mathtt{c_2}~(bind~\mathtt{o_1}~ \mathtt{o_2} ~\mathtt{M_1}~\mathtt{M_2})
\end{array}$
\end{theorem}
\begin{proof}
\pfsketch{
We prove that any two classes $c_1$ and $c_2$ linked by the $inh$ relation in any two models $M_1$ and $M_2$, are also linked by the $inh$ relation in the model obtained from the $bind$ of any two model's elements $o_1$ and $o_2$ using the two models $M_1$ and $M_2$. We suppose that we have a relation $inh$ between any two model's elements typed by $c_1$ and $c_2$ in the two models $M_1$ and $M_2$. We verify first that $o_1$ is a $Hook$ for the model $M_1$ and $o_2$ is a $Prototype$ for the model $M_2$ and that $o_1$ and $o_2$ have the same type otherwise the $bind$ returns the model $M_1$ and the proof is trivial. In case of $o_1$ is a $Hook$ and $o_2$ is a $Prototype$ and the two model's elements have the same type, we show that the $bind$ does not change the types of model's elements and the types of relations and so in the resulted model we have always an $inh$ relation between any model's elements typed by $c_1$ and $c_2$. The \coq proof is long but straightforward and considers all the cases of equality between the name of any model's element typed by $c_1$ and the names of the model's elements $o_1$ and $o_2$ and shows in all cases that the $inh$ relation is preserved.
}
\end{proof}

So, there is no necessary precondition on the parameters of the \bind operator to verify that the \texttt{subClass} property is compositional.

\paragraph{The isAbstract property:}
Abstract classes that are specified in a
metamodel using the \emph{isAbstract} attribute are
not suitable for instantiation. They are often used to represent
abstract concepts or entities. 

\[
\begin{array}{l}
isAsbstract(c_1 \in \texttt{Classes}, \langle MV, ME \rangle) \triangleq \\
\qquad \forall~o\in \texttt{Objects}, \langle o: c_1 \rangle \in MV \Rightarrow \exists~c_2\in \texttt{Classes}, \langle \langle o: c_2\rangle, inh, \langle o: c_1\rangle\rangle\in ME
\end{array}
\]

The preservation of this property by the \bind operator is proved\footnote{\url{http://coq4mde.enseeiht.fr/FormalMDE/Bind2M_Verif.html\#Bind2MIAP}} using the theorem~\ref{BindIsAbstractPreserved}. This theorem shows that all the abstract classes in any two models $\mathtt{M_1}$ and $\mathtt{M_2}$ are also abstract in the model obtained by the application of the \bind operator on the two models.

\begin{theorem}\emph{(BindIsAbstractPreserved)}\label{BindIsAbstractPreserved}\\
$\begin{array}{l}
\quad \forall~\mathtt{M_1}~\mathtt{M_2} \in \texttt{Model},~ \mathtt{c} \in \texttt{Classes},~ \mathtt{o_1}~\mathtt{o_2} \in \texttt{Objects},\\
\qquad (isAbstract~\mathtt{c}~\mathtt{M_1}) \wedge~(isAbstract~\mathtt{c}~\mathtt{M_2})
 \rightarrow isAbstract~ \mathtt{c}~(bind~\mathtt{o_1}~ \mathtt{o_2} ~\mathtt{M_1}~\mathtt{M_2})
\end{array}$
\end{theorem}
\begin{proof}
\pfsketch{
We prove that any abstract class $c$ in any two models $M_1$ and $M_2$, is also abstract in the model obtained from the $bind$ of any two model's elements $o_1$ and $o_2$ using the two models $M_1$ and $M_2$. We suppose that the class $c$ is abstract in the models $M_1$ and $M_2$. We verify first that $o_1$ is a $Hook$ for the model $M_1$ and $o_2$ is a $Prototype$ for the model $M_2$ and that $o_1$ and $o_2$ have the same types otherwise the $bind$ returns the model $M_1$ and the proof is trivial. In case of $o_1$ is a $Hook$ and $o_2$ is a $Prototype$ and the two model's elements have the same type, we show that the $bind$ does not change the types of model's elements and so in the resulted model if an element typed by the class $c$ is in the resulting model, then another element having the same name typed by $c_1$ and linked to the first model's element with an $inh$ relation will be also in the resulting model. The \coq proof is long and considers all the cases of equality between the name of any model's element typed by $c$ and the names of the model's elements $o_1$ and $o_2$ and shows in all cases that the relation is preserved.
}
\end{proof}

So, there is no precondition on the parameters of the \bind operator to verify that the \texttt{isAbstract} property is compositional.

\paragraph{The lower \& upper properties:}
A minimum and maximum number
of instances of target attribute or reference can be defined using the $lower$ and
$upper$ attributes. Both
attributes are used to represent a range of possible numbers of
instances. Unbounded ranges can be modelled using the $\top$ value for
the $upper$ attribute.
\[
\begin{array}{l}                                                                                                                                                                                                
lower(\mathtt{c_1}\in \mathtt{MMV}, \mathtt{r_1} \in \mathtt{MME}, \mathtt{n} \in Natural^{\top}) \triangleq \langle \mathtt{MV}, \mathtt{ME} \rangle \mapsto\\
	   \qquad \forall~ o \in \mathtt{Objects}, \langle \mathtt{o}: \mathtt{c_1} \rangle \in \mathtt{MV} \Rightarrow \vert \{ \mathtt{m_2} \in \mathtt{MV} \mid \langle \langle \mathtt{o}: \mathtt{c_1} \rangle, \mathtt{r_1}, \mathtt{m_2} \rangle \in \mathtt{ME} \} \vert \geq \mathtt{n}
\end{array} 
\]\label{lower}
The theorem~\ref{BindLowerPreserved} (\texttt{BindLowerPreserved}) shows\footnote{\url{http://coq4mde.enseeiht.fr/FormalMDE/Bind2M_Verif.html\#Bind2MLP}} that the  \texttt{lower} property is preserved by the \bind operator. The verification requires the \bind operator to be injective and preserves the difference between the elements in the resulting model. This is ensured if the model's element $\mathtt{o_2}$ is not in the first model, this verifies that the \bind operator does not add an element that already exists in the model. Finally, the preservation of the \texttt{lower} property is proven. An analogous formalization for the \texttt{lower} property is defined for the \texttt{upper} property
replacing $\geq$ by $\leq$.
\begin{theorem}\emph{(BindLowerPreserved)}\label{BindLowerPreserved}\\
$\begin{array}{l}
\quad			\forall~\langle \mathtt{MV_1},\mathtt{ME_1} \rangle~~ \langle \mathtt{MV_2},\mathtt{ME_2} \rangle \in \texttt{Model},~ \mathtt{c}  \in \texttt{Classes},~ \mathtt{r}   \in\ \texttt{Refer\-ences},\\
\quad			\mathtt{n} \in Natural^{\top},~ (\mathtt{o_1}: \mathtt{c_1})~(\mathtt{o_2}: \mathtt{c_2})\in \texttt{Objects}, \\
\qquad			\mathtt{c_1} = \mathtt{c_2} \wedge~(\forall~c, (\mathtt{o_2}: \mathtt{c}) \notin \mathtt{MV_1})
			\wedge~\textit{Injectif} ~bind\\
\qquad			\wedge~(lower~\mathtt{c}~\mathtt{r}~\mathtt{n}~\langle \mathtt{MV_1},\mathtt{ME_1} \rangle) \wedge~ (lower~\mathtt{c}~\mathtt{r}~\mathtt{n}~\langle \mathtt{MV_2},\mathtt{ME_2} \rangle)\\
\qquad			\rightarrow (lower~ \mathtt{c}~\mathtt{r}~\mathtt{n}~(bind ~(\mathtt{o_1}: \mathtt{c_1})~(\mathtt{o_2}: \mathtt{c_2})~\langle \mathtt{MV_1},\mathtt{ME_1} \rangle~\langle \mathtt{MV_2},\mathtt{ME_2} \rangle)).
\end{array}$
\end{theorem}

\begin{proof}
\pfsketch{
We suppose for any two models $\langle \mathtt{MV_1},\mathtt{ME_1} \rangle$ and $\langle \mathtt{MV_2},\mathtt{ME_2} \rangle$ that a lower bound $n$ is satisfied for the class $c$ in  relation with the reference $r$ (maximum $n$ model's elements are related by the relation $r$ to the same instance of the class $c$). Then, we prove that this lower bound $n$ is also satisfied in the model obtained from the $bind$ of any two model's elements $o_1$ and $o_2$ using the two models $\langle \mathtt{MV_1},\mathtt{ME_1} \rangle$ and $ \langle \mathtt{MV_2},\mathtt{ME_2} \rangle$. We verify first like in the previous proofs that $o_1$ is a $Hook$ for the model $\langle \mathtt{MV_1},\mathtt{ME_1} \rangle$ and $o_2$ is a $Prototype$ for the model $\langle \mathtt{MV_2},\mathtt{ME_2} \rangle$ and that $o_1$ and $o_2$ have the same types otherwise the $bind$ returns the model $M_1$ and the proof is trivial. In case of $o_1$ is a $Hook$ and $o_2$ is a $Prototype$ and the two model's elements have the same types, we show that the $bind$ does not change the types of the model's elements and does not reduce the lower bound in the resulting model because the $bind$ is supposed injective and so does not introduce new model's elements duplications. The \coq proof is long and uses intermediate lemmas to simplify iterations and calculations of the links and the model's elements (the difficulty is linked to the elegant coding of the  graphs and the models using dependent types). This proof considers also all the cases of equality between the name of any instance of $c$ and the names of the model's elements $o_1$ and $o_2$ and shows in all cases that the lower bound is preserved.}
\end{proof}

The preservation of the \texttt{upper} property is described\footnote{\url{http://coq4mde.enseeiht.fr/FormalMDE/Bind2M_Verif.html\#Bind2MUP}} by the \texttt{BindUpperPreserved} theorem which is similar to the previous theorem for the \texttt{lower} property.
So, we find it necessary to introduce assumptions about the model's elements to ensure that the composition using the \bind operator preserves the \texttt{lower} and \texttt{upper} properties. There are therefore preconditions on the \bind operator to ensure the preservation of these properties.

\paragraph{The isOpposite property:}
A reference can be associated to an $opposite$ reference.  It implies
that, in a valid model, for each link instance of this
reference between two objects, a link in the
opposite direction between the same objects exists.
\[
\begin{array}{l}                                                                                                                                                                                                
isOpposite(\mathtt{r_1}, \mathtt{r_2} \in \mathtt{MME}) \triangleq \langle \mathtt{MV}, \mathtt{ME} \rangle \mapsto
\forall ~ \mathtt{m_1}, \mathtt{m_2} \in \mathtt{MV}, \langle \mathtt{m_1}, \mathtt{r_1}, \mathtt{m_2} \rangle \in \mathtt{ME}                                                                                                                             
\Leftrightarrow \langle \mathtt{m_2}, \mathtt{r_2}, \mathtt{m_1} \rangle \in \mathtt{ME}                                                                                                                      
\end{array}                                                                                                                                                                                                     
\]
The theorem~\ref{BindIsOppositePreserved} (\texttt{BindIsOppositePreserved}) shows\footnote{\url{http://coq4mde.enseeiht.fr/FormalMDE/Bind2M_Verif.html\#Bind2MIOP}} that each pair of  opposite references in the two models $M_1$ and $M_2$ are also opposite in the resulting model from applying the \bind operator on the two models.
Finally, the property \texttt{isOpposite} is preserved.

\begin{theorem}\emph{(BindIsOppositePreserved)}\label{BindIsOppositePreserved}\\
$\begin{array}{l}
\quad \forall~\mathtt{M_1}~\mathtt{M_2} \in \texttt{Model},~ \mathtt{r_1}~\mathtt{r_2} \in \texttt{Refer\-ences},~ \mathtt{o_1}~\mathtt{o_2} \in \texttt{Objects}, \\
\qquad (isOpposite~\mathtt{r_1}~\mathtt{r_2}~\mathtt{M_1}) \wedge~(isOpposite~\mathtt{r_1}~\mathtt{r_2}~\mathtt{M_2})
\rightarrow (isOpposite~\mathtt{r_1}~\mathtt{r_2} ~(bind~\mathtt{o_1}~\mathtt{o_2}~\mathtt{M_1}~\mathtt{M_2})).
\end{array}$
\end{theorem}

\begin{proof}
\pfsketch{
We prove that any two references $r_1$ and $r_2$ that are opposite in any two models $M_1$ and $M_2$, are also opposite in the model obtained from the $bind$ of any two model's elements $o_1$ and $o_2$ using the two models $M_1$ and $M_2$. We verify first like in all the other proofs that $o_1$ is a $Hook$ for the model $M_1$ and $o_2$ is a $Prototype$ for the model $M_2$ and that $o_1$ and $o_2$ have the same type otherwise the $bind$ returns the model $M_1$ and the proof is trivial. In case of $o_1$ is a $Hook$ and $o_2$ is a $Prototype$ and the two model's elements have the same type, we show that the $bind$ does not change the references and so we can find all the opposite references from the initial models. The \coq proof considers all the cases of equality between the names of the model's elements and the names of $o_1$ and $o_2$ and shows in all cases the preservation of the opposite references.
}
\end{proof}

So, there is no precondition on the parameters of the \bind operator to verify that the \texttt{isOpposite} property is compositional.

\paragraph{The areComposite property:}
A reference can be $composite$ and, as a matter of fact, defining a set
of references considered as a whole to be composite, instead of a
single one, appears closer to the intended meaning. In such a case,
instances of the target concept belong to a single instance of source
concepts.
\[   
\begin{array}{l}                                                                                                                                                                                  
areComposite(\mathtt{c_1} \in \mathtt{MMV}, R \subseteq \mathtt{MME}) \triangleq \langle \mathtt{MV}, \mathtt{ME} \rangle \mapsto\\
\qquad \forall~ o \in \mathtt{Objects} \Rightarrow 
\vert \{ \mathtt{m_1} \in \mathtt{MV} \mid \langle \mathtt{m_1}, \mathtt{r}, \langle \mathtt{o}: \mathtt{c_1} \rangle\rangle \in \mathtt{ME}, \mathtt{r} \in R \} \vert \leq 1                                                                         
\end{array}
\]
The theorem~\ref{BindareCompositeSubsPreserved} (\texttt{BindAreCompositeSubsPreserved}) shows\footnote{\url{http://coq4mde.enseeiht.fr/FormalMDE/Bind2M_Verif.html\#Bind2MACP}} that the set of composite references in the two models $\mathtt{M_1}$ and $\mathtt{M_2}$ are also composite in the resulted model from the application of the \bind operator on the two models. 
This theorem requires also that the \bind operator is injective and requires that the substituted model does not contain an element whose name is $\mathtt{o_2}$.

\begin{theorem}\emph{(BindAreCompositeSubsPreserved)}\label{BindareCompositeSubsPreserved}\\
$\begin{array}{l}
\quad			\forall~\langle \mathtt{MV_1},\mathtt{ME_1} \rangle~\langle \mathtt{MV_2},\mathtt{ME_2} \rangle \in \texttt{Model}, ~\mathtt{c} \in \texttt{Classes}, \\
\quad			\mathtt{R}  \subset \texttt{Refer\-ences},~\mathtt{o_1} ~\mathtt{o_2} \in \texttt{Objects},~\mathtt{c_2} ~\mathtt{c_2} \in \texttt{Classes},\\
\qquad			\mathtt{c_1} = \mathtt{c_2}~\wedge~(\forall~c, (\mathtt{o_2}: \mathtt{c}) \notin \mathtt{MV_1})
			\wedge~\textit{Injectif} ~bind\\
\qquad			\wedge~(areComposite~\mathtt{c}~\mathtt{R}~\langle \mathtt{MV_1},\mathtt{ME_1} \rangle) \wedge~(areComposite~\mathtt{c}~\mathtt{R}~\langle \mathtt{MV_2},\mathtt{ME_2} \rangle)\\
\qquad 			\rightarrow (areComposite~\mathtt{c}~\mathtt{R}~(bind~(\mathtt{o_1}: \mathtt{c_1})~(\mathtt{o_2}: \mathtt{c_2})~\langle \mathtt{MV_1},\mathtt{ME_1} \rangle~\langle \mathtt{MV_2},\mathtt{ME_2} \rangle))
\end{array}$
\end{theorem}

\begin{proof}
\pfsketch{
We suppose for any two models $\langle \mathtt{MV_1},\mathtt{ME_1} \rangle$ and $\langle \mathtt{MV_2},\mathtt{ME_2} \rangle$, for any instance of a class $c$ in these models, at most one ancestor is linked with a composite reference. We verify that this property is also satisfied in the model obtained from the $bind$ of any two model's elements $o_1$ and $o_2$ using these two models. We verify first like in the previous proofs that $o_1$ is a $Hook$ for the model $\langle \mathtt{MV_1},\mathtt{ME_1} \rangle$ and $o_2$ is a $Prototype$ for the model $\langle \mathtt{MV_2},\mathtt{ME_2} \rangle$ and that $o_1$ and $o_2$ have the same type otherwise the $bind$ returns the model $M_1$ and the proof is trivial. In case of $o_1$ is a $Hook$ and $o_2$ is a $Prototype$ and the two model's elements have the same type, we show that the $bind$ does not change the types of the model's elements and does not increase the number of composite references for any model's element and this by supposing like for the $lower$ property proof that the $bind$ is injective and so does not introduce new model's elements duplications. The \coq proof is long and uses intermediate lemmas to simplify iterations and calculations of the references and models' elements (the difficulty is linked to elegant coding of graphs and models using dependent types). The proof considers also all the cases of equality between the name any instance of $c$ in the two models and the names of the model's elements $o_1$ and $o_2$ and shows in all cases that the limit for the number of composite relation for any model's element is preserved.}
\end{proof}

\subsection{The \bind operator with several variation points}
This version is a generalization of the \bind operator. It is characterized by a list $l$ of variation and reference points (\bind of two Models with Several Hooks). 

$\mathtt{Bind2MSH}: Model~\times~Model~\times~list ~(Objects~\times~Objects)$ is defined as:
\[
\begin{array}{l}
\mathtt{Bind2MSH}~\mathtt{M_1} ~\mathtt{M_2}~l = 
\forall~(o,o') \in l, bind ~\mathtt{M_1} ~\mathtt{M_2}~o~o'~l
\end{array}
\]
The proofs of properties require the following assumptions: type compatibility between two model's elements for each pair of elements in the list, the \bind operator to be injective and an additional condition: the same \texttt{Prototype} is not given more than once to ensure the preservation of multiplicities. The same assumptions/preconditions are necessary to prove the compositional verification of the various considered properties.

The proofs for this version of the \bind operator use the proofs of the \bind basic operator in addition to a standard schema to find the target model and the application conditions. The language of tactics for the \coq system~\cite{delahaye2000tactic} is used to define the tactics that significantly improved the proofs\footnote{\url{http://coq4mde.enseeiht.fr/FormalMDE/Bind2M_Verif.html}}.
     
\subsection{The \texttt{extend} operator}\label{extOp}
Two variations of the \texttt{extend} operator are implemented. The first version makes the hypothesis in addition to the \texttt{extend} operator definition that the two models are disjoint to define the predicate $extensibleC$.

The second version does not make any assumptions about the intersection of models. In this latest version, models can contain common elements as they may result from the extraction of components from the same model. We do not present the proofs of the \mof properties for these operators in this paper, but the proofs are finalized and the interested reader can refer to the special page\footnote{\url{http://coq4mde.enseeiht.fr/FormalMDE/Extend_Verif.html}}.
		
This work presents the preconditions allowing for each operator (\isc basic operators) to generate consistent metamodels. Detecting and resolving the conflicts require the compositional application of several composition operators (each operator is proved preserving the properties) and contributes for the satisfaction of the next applied operator preconditions. For example, this can be used to find a sensible unification of the constraints contributed by the two model's fragments being composed.

\section{Related work}
In the first version of the \isc composition
method~\cite{assmann2003invasive}, the notion of conformance is
restricted to the $instanceOf$ property defined
in~\cite{kezadri2011proof}. A composition operator is safe if it
preserves the consistence (Theorem 5.1 (Sound Composition retains
Consistency) in \cite{assmann2003invasive}). The first version of \isc
was defined on fragments of Java code, the extension operator
guarantees by definition that it will not change the code of a fragment box,
although it can change its semantics. The semantics is preserved if the
added code to the variation point is independent of the code of the fragment box (Theorem 5.2
(Sound Composition with Extension Composers)
in~\cite{assmann2003invasive}). We proved that the semantics is
preserved if the models are disjoint (Section \ref{extOp}). Moreover,
we extended this work by offering the preconditions that enable to
preserve the semantics and all the formal proofs in the \coq proof
assistant.

In the last version of \isc~\cite{johannes2010component} implemented
in the \reuseware framework as an Eclipse plugin and developed in
parallel with our work, the typing property is ensured in relation
with some properties of the \mof metamodel using the notion of
compatibility between the variation and reference points. But and as
presented in the motivating example of this paper, this version does
not take into account all the semantics properties of the \mof
metamodel and inconsistent metamodels can be generated by composition.
We presented then the theorems proving the preservations of some of
the \mof semantics properties and the preconditions for the
verification.

Our approach is original compared to the work of A{\ss}man~\cite{assmann2003invasive}, we provide
in advance the preconditions which ensure that the result of applying
an operator is valid (typing and semantics properties). We do not need
to check for each application that the result is valid, but we know
the preconditions that must be met and if our conditions are
satisfied, we can ensure that the result of the composition is
consistent. The expected direct consequences for our work are: the use of \coqformde to prove the correction of the \isc method itself and the composition methods in general by introducing and proving the preconditions ensuring the properties preservation.

This work is also closely related to all work about the formalization of model driven engineering, we present first in what follows some approaches based on shallow encoding and then compare them to our formalization. We present finally briefly a deep encoding for the \mde concepts associated with a highlight for the differences with our encoding.

MoMENT (MOdel manageMENT) \cite{DBLP:journals/fac/BoronatM10} is an algebraic model management framework that provides a set of generic
operators to manipulate models. In the MoMENT framework, the metamodels are represented as algebraic
specifications and the operators are defined independently of the metamodel using the \maude language \cite{clavel2002maude}. To be used, the operators must be specified in a module called signature that specifies the constructs of the metamodel. The approach was implemented in a tool\footnote{
  \url{http://moment.dsic.upv.es/}} that gives also an automatic
translation from an \texttt{EMF} metamodel to a signature model.

A. Vallecillo et al. have designed and implemented a
different embedding of metamodels, models
(\cite{DBLP:journals/jot/RomeroRDV07}) and model
transformations (\cite{DBLP:conf/icmt/TroyaV10}) using
\maude. This embedding relies on the object rewriting semantics in
order to implement model transformations.

I. Poernomo has proposed an encoding of metamodels and
models using type theory (\cite{DBLP:conf/sac/Poernomo06}) in
order to allow correct by construction development of model
transformation using proof-assistants like \coq
(\cite{DBLP:conf/icmt/Poernomo08}). Some simple experiments have been
conducted using \coq mainly on tree-shaped models
(\cite{DBLP:conf/icfem/PoernomoT10}) using inductive types. General
\texttt{graph} model structure can be encoded using co-inductive
types. However, as shown in \cite{2011.314} by C. Picard and
R. Matthes, the encoding is quite complex as \coq enforces structural
constraints when combining inductive and co-inductive types that
forbid the use of the most natural encodings proposed by Poernomo et
al. M. Giorgino et al. rely in \cite{giorgino2011verification} on a spanning tree of
the graph combined with additional links to overcome that
constraint using the \isabelle proof-assistant. This allows to develop
a model transformation relying on slightly adapted inductive proofs
and then extract classical imperative implementations. 

The HOL-OCL system \cite{brucker2002proposal} \cite{brucker2008hol} is an environment for interactive modelling with \uml and \ocl that can be used for example to prove class invariants. 

These
embeddings are all shallow: they rely on sophisticated similar data
structure to represent model elements and metamodels
(e.g. \coq (co-)inductive data types for model elements and object
and (co-)inductive types for metamodel elements).

The work described in this paper is a deep embedding, each concept
from models and metamodels was encoded in \cite{towers07} using
elementary constructs instead of relying on similar elements in
\maude, \coq or \isabelle. The purpose of this contribution is
not to implement model transformation using correct-by-construction
tools but to give a kind of denotational semantics for model-driven
engineering concepts that should provide a deeper understanding and
allow the formal validation of the various implemented technologies.
Other work aiming to define a semantics for a modelling language by explicitly and denotationally define the kind of systems the language describes and to focus on the variations and variability in the semantics \cite{cengarle2008system} \cite{maoz2011semantically}. Compared to the last work, we are interested in a complete and unique formalisation of the conformity to metamodels, of course this property must be considered in the more general consistency relation and we are focused mainly in the proof of the preservation of this conformity relation by the \isc composition operators. 

Another formalisation in \coq of the \mde concepts by F.Barbier et al is accessible\footnote{\url{http://web.univ-pau.fr/~barbier/Coq/}}~\cite{barbierhal00840748}, this representation is attached to the proof of the properties shown in~\cite{kuhne2006matters} (on instantiation relations and model transformations). The last formalization differs from ours by a detailed representation of the different components of models and metamodels based on the \mof concepts. The \coqformde formalisation has the advantage to be more generic and minimum through the use of modules for the representation of these concepts and its support for a large variety of properties describing the conformity by a predicate integrated to the metamodel type.

%---------------------------------------------------------------------------------------------

\section{Conclusion}

We have addressed the problem of compositional verification for models
relying on the generic composition method \isc and the \reuseware
framework. We first proposed in~\cite{kezadri2011proof} a formalization for the \isc composition method and the
verification of type safety for these operators. Then, we presented in
this paper the verification of generic semantics properties in
relation with the \mof metametamodel.

This integration enables to extract executable correct by construction
composition operators. The termination of the extracted operators is
ensured by the \coq definition. The typing property and a set of
semantics properties in relation with the \mof metametamodel are
proved preserved directly or by the composition operators by
introducing some preconditions on the parameters of the composition
operators. The application is not limited to a specific language, but
can be extended to all models and modeling languages defined by
metamodels. From the \isc composition method basic operators (\bind and
\texttt{extend}), more complex operators were built. The complex operators
allow more complex transformations such as linking several variation
points at the same time.

This proposal is a required step in the formalization of compositional
verification techniques. The next step of our work is to formalize
other composition operators and to take into account others static
constraints such as \texttt{OCL} constraints~\cite{omg2012ocl} and
more dynamic properties such as the deadlock freedom proposed in the
BIP framework~\cite{basu2006modeling}. The expected result of our work
is to define a correct by construction framework for combining several
composition techniques.

\bibliographystyle{eptcs}
\bibliography{FESCA2014}  % sigproc.bib is the name of the Bibliography in this case

\appendix
\vspace{-0.7cm}
\section{A part from the ValidBind theorem proof}
\vspace{-0.4cm}
This appendix presents the mathematical proof for the first theorem presented in the Section~\ref{V}. The theorem \emph{ValidBind} proves the preservation of the $instanceOf$ property by the \texttt{bind} operator. The Lamport's method \cite{lamport1995write} is used to write this proof. 
%We don't present to other proofs in the same format but we reference for every theorem the link to its \coq proof in a special web site.

\begin{theorem}\emph{(ValidBind)}\label{ValidBind}\\
$~~~~~~~InstanceOf~  (\mathtt{M_1}, \mathtt{MM}) ~ \wedge ~InstanceOf ~ (\mathtt{M_2}, \mathtt{MM}) 
\rightarrow InstanceOf ~ ((bind  ~ \mathtt{o_1}~\mathtt{o_2}~\mathtt{M_1}  ~ \mathtt{M_2}), \mathtt{MM})$
\end{theorem}

\begin{proof}

\assume{$M_1$ the $\model \langle \mathtt{MV}, \mathtt{ME} \rangle$\\
		$M_2$ the $\model \langle \mathtt{MV_1}, \mathtt{ME_1} \rangle$\\
		 $MM$ the $\metamodel \langle \mathtt{MMV}, \mathtt{MME} ,conformsTo \rangle$\\
		 $H$: $InstanceOf~(\langle \mathtt{MV}, \mathtt{ME} \rangle, \langle \mathtt{MMV}, \mathtt{MME} ,conformsTo \rangle)$.\\
		 $H_{M2}$: $InstanceOf~(\langle \mathtt{MV_1}, \mathtt{ME_1} \rangle, \langle \mathtt{MMV}, \mathtt{MME} ,conformsTo \rangle)$.
}

\prove{
$InstanceOf ~ ((bind  ~ \mathtt{o_1}~\mathtt{o_2}~\langle \mathtt{MV}, \mathtt{ME} \rangle  ~ \langle \mathtt{MV_1}, \mathtt{ME_1} \rangle), \langle \mathtt{MMV}, \mathtt{MME} ,conformsTo \rangle)$
}

\pfsketch{
We suppose that the two models $M_1$ and $M_2$ are instance of the metamodel $MM$  and we prove that the model obtained by applying the $bind$ operator on the two models using two model's elements $o_1$ and $o_2$ is also instance of the metamodel $MM$. We verify first that $o_1$ is a $Hook$ for the model $M_1$ and $o_2$ is a $Prototype$ for the model $M_2$ and that $o_1$ and $o_2$ have the same types otherwise the $bind$ returns the model $M_1$ and the proof is trivial. In case of $o_1$ is a $Hook$ and $o_2$ is a $Prototype$ and the two model's elements have the same type (the case detailed below), we show that the $bind$ does not change the types of the vertices and edges and so preserves the type safety.
}

\pf

\step{<1>}{After introducing the definitions of $instanceOf$ and the $bind$ operator, the hypothesis $H$ becomes:

$H$: $(\forall \langle \mathtt{o},\mathtt{c} \rangle, \langle \mathtt{o},\mathtt{c} \rangle \in \mathtt{MV} \rightarrow \mathtt{c} \in \mathtt{MMV})\\
 \wedge
(\forall \langle \langle \mathtt{o},\mathtt{c} \rangle , \mathtt{r}, \langle \mathtt{o'},\mathtt{c'} \rangle \rangle, \langle \langle \mathtt{o},\mathtt{c} \rangle , \mathtt{r}, \langle \mathtt{o'},\mathtt{c'} \rangle \rangle \in \mathtt{ME} \rightarrow \langle \mathtt{c},\mathtt{r},\mathtt{c'} \rangle \in \mathtt{MME})$.

The current goal is transformed:

$(\forall \langle \mathtt{o},\mathtt{c} \rangle, \langle \mathtt{o},\mathtt{c} \rangle \in ( V.image~mapv~(\langle \mathtt{MV}, \mathtt{ME} \rangle )~g) \rightarrow \mathtt{c} \in \mathtt{MMV})\\
 \wedge
(\forall \langle \langle \mathtt{o},\mathtt{c} \rangle , \mathtt{r}, \langle \mathtt{o'},\mathtt{c'} \rangle \rangle, \langle \langle \mathtt{o},\mathtt{c} \rangle , \mathtt{r}, \langle \mathtt{o'},\mathtt{c'} \rangle \rangle \in (E.image~mapv~mapa~(\langle \mathtt{MV}, \mathtt{ME} \rangle )~g)\\
 \rightarrow \langle \mathtt{c},\mathtt{r},\mathtt{c'} \rangle \in \mathtt{MME})$.}

\step{<2>}{The hypothesis $H$ is divided into two hypotheses:

$H_0$:
$\begin{array}{l}\label{H0}
\forall \langle \mathtt{o},\mathtt{c} \rangle, \langle \mathtt{o},\mathtt{c} \rangle \in \mathtt{MV} \rightarrow \mathtt{c} \in \mathtt{MMV}
\end{array}$.

$H_1$:
$\begin{array}{l}\label{H1}
\forall \langle \langle \mathtt{o},\mathtt{c} \rangle , \mathtt{r}, \langle \mathtt{o'},\mathtt{c'} \rangle \rangle, \langle \langle \mathtt{o},\mathtt{c} \rangle , \mathtt{r}, \langle \mathtt{o'},\mathtt{c'} \rangle \rangle \in \mathtt{ME} \rightarrow \langle \mathtt{c},\mathtt{r},\mathtt{c'} \rangle \in \mathtt{MME}
\end{array}$.

The current goal is divided into two sub-goals:
\begin{pfenum}
\item $\langle \mathtt{o},\mathtt{c} \rangle \in ( V.image~mapv~(\langle \mathtt{MV}, \mathtt{ME} \rangle )~g) \rightarrow \mathtt{c} \in \mathtt{MMV}$

\item $\forall \langle \langle \mathtt{o},\mathtt{c} \rangle , \mathtt{r}, \langle \mathtt{o'},\mathtt{c'} \rangle \rangle, \langle \langle \mathtt{o},\mathtt{c} \rangle , \mathtt{r}, \langle \mathtt{o'},\mathtt{c'} \rangle \rangle \in (E.image~mapv~mapa~(\langle \mathtt{MV}, \mathtt{ME} \rangle )~g)\\
 \rightarrow \langle \mathtt{c},\mathtt{r},\mathtt{c'} \rangle \in \mathtt{MME}$
\end{pfenum}
}

\begin{proof}

\step{<3>}{We begin by proving the first subgoal that corresponds to the left side of the conjunction:

\assume{ $H_2$: 
$\begin{array}{l}
\langle \mathtt{o},\mathtt{c} \rangle \in ( V.image~mapv~(\langle \mathtt{MV}, \mathtt{ME} \rangle )~g)
\end{array}$.
}

\prove{$\begin{array}{l}
(\mathtt{c} \in \mathtt{MMV})
\end{array}$}

\pf

\begin{proof}
\step{<3>1}{By generalizing the lemma~\ref{VImageElim} using $H_2$, we get a new hypothesis:
$H_4$:
$\exists (o', c') \in MV \mid mapv~(o', c') = (o, c)$.}

\step{<3>2}{We introduce the definition of $mapv$, we can conclude that $c'=c$

then we have as hypothesis: $H_5: (o', c) \in MV$.}

\step{<3>3}{By applying $H_0$ with as parameter $(o',c)$ and $H_5$.}

\qedstep

\end{proof}
}

\step{4}{We now prove the second part of the goal:

$\begin{array}{l}
\forall \langle \langle \mathtt{o},\mathtt{c} \rangle , \mathtt{r}, \langle \mathtt{o'},\mathtt{c'} \rangle \rangle, \langle \langle \mathtt{o},\mathtt{c} \rangle , \mathtt{r}, \langle \mathtt{o'},\mathtt{c'} \rangle \rangle \in (E.image~mapv~mapa~(\langle \mathtt{MV}, \mathtt{ME} \rangle )~g)\\
 \rightarrow \langle \mathtt{c},\mathtt{r},\mathtt{c'} \rangle \in \mathtt{MME}
\end{array}$

\begin{proof}

\assume{Having as an additional hypothesis to $H_0$ and $H_1$, the hypothesis $H_2$:
$
\forall \langle \langle \mathtt{o},\mathtt{c} \rangle , \mathtt{r}, \langle \mathtt{o'},\mathtt{c'} \rangle \rangle, \\
\langle \langle \mathtt{o},\mathtt{c} \rangle , \mathtt{r}, \langle \mathtt{o'},\mathtt{c'} \rangle \rangle \in (E.image~mapv~mapa~(\langle \mathtt{MV}, \mathtt{ME} \rangle )~g)$}

\prove{This sub-goal can be resolved by proving:$\begin{array}{l}
\langle \mathtt{c},\mathtt{r},\mathtt{c'} \rangle \in \mathtt{MME}
\end{array}$}

\pf

\step{4.1}{Here, we generalize the lemma~\ref{EImageElim} using the hypothesis $H_2$, we get a new hypothesis:
$H_4$:
$\exists \langle \langle \mathtt{o_1},\mathtt{c_1} \rangle , \mathtt{r_1}, \langle \mathtt{o_1'},\mathtt{c_1'} \rangle \rangle,\\
 \langle \langle \mathtt{o_1},\mathtt{c_1} \rangle , \mathtt{r_1}, \langle \mathtt{o_1'},\mathtt{c_1'} \rangle \rangle \in \mathtt{ME} \mid mape~\langle \langle \mathtt{o_1},\mathtt{c_1} \rangle , \mathtt{r_1}, \langle \mathtt{o_1'},\mathtt{c_1'} \rangle \rangle = \langle \langle \mathtt{o},\mathtt{c} \rangle , \mathtt{r}, \langle \mathtt{o'},\mathtt{c'} \rangle \rangle$.}

\step{4.2}{We introduce the definition of $mape$, we can conclude that:\\
 $c_1=c$, $c_1'=c'$ et $r_1=r$, 

then we have as hypothesis: $H_5: \langle \langle \mathtt{o_1},\mathtt{c} \rangle , \mathtt{r}, \langle \mathtt{o_1'},\mathtt{c'} \rangle \rangle \in \mathtt{ME}$.}

\step{4.3}{By applying $H_0$ with as parameter $\langle \langle \mathtt{o_1},\mathtt{c} \rangle , \mathtt{r}, \langle \mathtt{o_1'},\mathtt{c'} \rangle \rangle$ and $H_5$.}

\qedstep

\end{proof}

}

\qed

\end{proof}

\end{proof}
Lemmas used in this proof are:

\begin{lemma}\emph{(V.imageElim)}\label{VImageElim}\\
$\begin{array}{l}
\forall~mapv,~mapa,~\langle \mathtt{MV}, \mathtt{ME} \rangle \in Model ,
v \in (image~\langle \mathtt{MV}, \mathtt{ME} \rangle)
\rightarrow \exists w \in MV \wedge mapv~w = v.
\end{array}$
\end{lemma}\label{V.imageElim}

\begin{lemma}\emph{(E.imageElim)}\\
$\begin{array}{l}
\forall~mapv,~mapa,~\langle \mathtt{MV}, \mathtt{ME} \rangle \in Model ,
e \in (image~\langle \mathtt{MV}, \mathtt{ME} \rangle)
\rightarrow \exists w \in \mathtt{ME} \wedge mape~w = e.
\end{array}$
\end{lemma}\label{EImageElim}
The proofs of these two lemmas are constructed by induction on the structure of the graph and involve other theorems that are not presented here but are available with our \coq code.
\end{document}